\documentclass{aims_LENCOS}
\usepackage{amsmath}
  \usepackage{paralist}
  \usepackage{graphics} 
  \usepackage{epsfig} 
 \usepackage[colorlinks=true]{hyperref}
\hypersetup{urlcolor=blue, citecolor=red}

\def\currentvolume{X}
 \def\currentissue{X}
  \def\currentyear{200X}
   \def\currentmonth{XX}
    \def\ppages{X--XX}

\newtheorem{theorem}{Theorem}[section]

\theoremstyle{definition}

\newtheorem{remark}{Remark}

\def\ad{\mbox{ad\,}}

\def\diag{\mbox{diag\,}}
\def\fr#1{{\mathfrak{#1}}}

\def\openone{\leavevmode\hbox{\small1\kern-3.3pt\normalsize1}}

\def\ad{\mathrm{ad\,}}

\def\bdiag{\mbox{block-diag\,}}
\def\diag{\mbox{diag\,}}

\def\bbbe\mathbb{E}
\def\bbbr{\mathbb{R}}

\def\bbbc{\mathbb{C}}
\def\bbbe{\mathbb{E}}
\def\bbbz{\mathbb{Z}}

\def\ad{\mbox{ad\,}}

\title[BEC and spectral properties of MNLS ]
      {Bose-Einstein Condensates and spectral properties of
      multicomponent nonlinear Schr\"odinger equations }

\author[Vladimir S. Gerdjikov]
{}


 \email{gerjikov@inrne.bas.bg}

\begin{document}
\maketitle

\centerline{\scshape Vladimir S. Gerdjikov }
\medskip
{\footnotesize
 \centerline{Institute for Nuclear Research and Nuclear Energy,  }
   \centerline{Bulgarian academy of sciences}
      \centerline{72 Tsarigradsko chaussee, 1784 Sofia, Bulgaria}
 } 


\medskip

\bigskip


\begin{abstract}
We analyze the properties of the soliton solutions of a class of models describing
one-dimensional BEC with spin $F$. We describe the minimal sets of scattering data
which determine uniquely both the corresponding potential of the Lax operator and
its scattering matrix. Next we give several reductions of these MNLS,  derive their
$N$-soliton solutions and analyze the soliton interactions. Finally we prove an important theorem
proving that if the initial conditions satisfy the reduction then one gets a solution
of the reduced MNLS.

\end{abstract}


\section{INTRODUCTION}
\label{sec:intro}  

It is well known that Bose-Einstein condensate (BEC) of alkali atoms in the $F=1$ hyperfine state,
elongated in $x$ direction and confined in the transverse directions $y,z$ by purely optical means are described
by a  3-component normalized spinor wave vector  $ {\bf\Phi}(x,t)=(\Phi_1, \Phi_0 , \Phi_{-1})^{T}(x,t)$.
Considering dimensionless units and using special choices for the scattering lengths one can show that ${\bf\Phi}(x,t)$
satisfies the multicomponent nonlinear Schr\"{o}dinger (MNLS) equation \cite{IMW04},
 see also \cite{imww04,LLMML05,uiw06,dwy08,Kevre*08}:
\begin{eqnarray}\label{eq:1}
&& i\partial_{t} \Phi_{1}+\partial^{2}_{x} \Phi_{1}+2(|\Phi_{1}|^2
+2|\Phi_{0}|^2) \Phi_{1} +2\Phi_{-1}^{*}\Phi_{0}^2=0, \nonumber \\
&& i\partial_{t} \Phi_{0}+\partial^{2}_{x}
\Phi_{0}+2(|\Phi_{-1}|^2
+|\Phi_{0}|^2+|\Phi_{1}|^2) \Phi_{0} +2\Phi_{0}^{*}\Phi_{1}\Phi_{-1}=0,\\
&& i\partial_{t}\Phi_{-1}+\partial^{2}_{x}
\Phi_{-1}+2(|\Phi_{-1}|^2+ 2|\Phi_{0}|^2) \Phi_{-1}
+2\Phi_{1}^{*}\Phi_{0}^2=0. \nonumber
\end{eqnarray}
Similarly spinor BEC with $F=2$ is described by a  5-component normalized spinor wave vector
$ {\bf\Phi}(x,t)=(\Phi_2,\Phi_1, \Phi_0 , \Phi_{-1} , \Phi_{-2})^{T}(x,t)$.
For  specific choices of the scattering lengths in dimensionless coordinates the corresponding set of equations for
${\bf\Phi}(x,t)$ take the form \cite{uiw07}:
\begin{eqnarray}\label{eq:F2}
&&i\partial_t\Phi_{\pm 2}+\partial_{xx}\Phi_{\pm 2} +2 (\vec{{\bf \Phi}},\vec{{\bf \Phi^{*}}}) \Phi_{\pm 2} -
(2\Phi_{2}\Phi_{-2}-2 \Phi_{1}\Phi_{-1}+\Phi_{0}^2) \Phi_{\mp 2}^*=0,\nonumber \\
&&i\partial_t\Phi_{\pm 1}+\partial_{xx}\Phi_{\pm 1} +2 (\vec{{\bf \Phi}},\vec{{\bf \Phi^{*}}}) \Phi_{\pm 1} +
(2\Phi_{2}\Phi_{-2}-2 \Phi_{1}\Phi_{-1}+\Phi_{0}^2) \Phi_{\mp 1}^*=0,\\
&&i\partial_t\Phi_{0}+\partial_{xx}\Phi_{0} +2(\vec{{\bf \Phi}},\vec{{\bf \Phi^{*}}}) \Phi_{0} -
(2\Phi_{2}\Phi_{-2}-2 \Phi_{1}\Phi_{-1}+\Phi_{0}^2)\Phi_{0}^*=0.\nonumber
\end{eqnarray}

Both models have natural Lie algebraic interpretation and are related to the symmetric
spaces ${\bf BD.I}\simeq {\rm SO(n+2)}/{\rm SO(n)\times SO(2)}$ with $n=3$ and $n=5$ respectively.
They are integrable by means of inverse scattering transform
method \cite{ForKu*83,tw98,GKV*09}. Using a modification of the Zakharov-Shabat
`dressing method' we describe the  soliton solutions \cite{IMW04,kagg07} and the effects of the reductions on them.

Sections 2  contains the basic details on the direct and inverse scattering problems for the Lax
operator. Section 3 is devoted to the construction of their soliton solutions. In Section 4 we formulate the minimal
sets of scattering data $L$ which determine uniquely both the scattering matrix and the potential $Q(x,t)$.
Section 5 gives a few important examples of algebraic reductions of the MNLS. In Section 6  we analyze the soliton
interactions of the MNLS. To this end we evaluate the limits of the generic  two-soliton solution
for $t\to\pm\infty$. As a result we establish that the effect of the interactions on the soliton
parameters is analogous to the one for the scalar NLS equation and consists in shifts of the `center of mass'
and shift in the phase. In Section 7  we prove an important theorem proving that if the initial
conditions satisfy the reduction then one gets a solution of the reduced MNLS.

\section{The method for solving MNLS for any $F$}
\subsection{The Lax representation}

The above MNLS equations (\ref{eq:1}) and (\ref{eq:F2}) are the first two members of a series of MNLS
equations related to the {\bf BD.I}-type symmetric spaces. They allow Lax representation
as follows \cite{ForKu*83,GKKV*08,GKV*09}
\begin{equation}\label{eq:3.1}
\begin{aligned}
L\psi (x,t,\lambda ) &\equiv  i\partial_x\psi + U(x,t,\lambda)\psi  (x,t,\lambda )=0, &\quad \\
M\psi (x,t,\lambda ) &\equiv  i\partial_t\psi + V(x,t,\lambda)\psi  (x,t,\lambda )=0, &\quad
\end{aligned}
\end{equation}
where
\begin{equation}\label{eq:3.1'}
\begin{aligned}
U(x,t,\lambda) &= Q(x,t) - \lambda J , \\
V(x,t,\lambda) &= V_0(x,t) +\lambda V_1(x,t) - \lambda ^2 J, \\
V_1(x,t) &= Q(x,t), \qquad V_0(x,t) = i \ad_J^{-1} \frac{d Q}{dx}
+ \frac{1}{2} \left[\ad_J^{-1} Q, Q(x,t) \right].
\end{aligned}
\end{equation}
For those familiar with Lie algebras I remind that, as usual, $Q(x,t)$ and $J$ are elements of the corresponding
Lie algebra, which in our case is $\mathfrak{g}\simeq so(n+2)$. The choice of the Cartan  subalgebra element $J$
determines the co-adjoint orbit of $\mathfrak{g}$; in our case $J$ is dual to $e_1$, see \cite{Helg}. It introduces grading in
$\mathfrak{g}= \mathfrak{g}^{(0)}\oplus \mathfrak{g}^{(1)}$ where $\mathfrak{g}^{(0)}\simeq so(n)$. The root system of
$\mathfrak{g}^{(0)}$ consists of all roots of $so(n+2)$ which are orthogonal to $e_1$; the linear subspace $\mathfrak{g}^{(1)}$
is spanned by the Weyl generators $E_\alpha$ and $E_{-\alpha}$ for which the roots  $\alpha\in\Delta_1^+$
are such that their scalar products $(\alpha,e_1)=1$. Thus the potential
\begin{equation}\label{eq:Q}
Q(x,t) = \sum_{\alpha \in \Delta_1^+} (q_\alpha(x,t) E_\alpha + p_\alpha(x,t) E_{-\alpha })
\end{equation}
may be viewed as local coordinate of the above mentioned symmetric space.
The linear operator $\ad_J X=[J,X]$ and $\ad_J^{-1}$ is well defined  on the image of $\ad_J$ in $\fr{g}$.

In what follows we will use the typical representation of $so(n+2)$ with $n=2r-1$ in which $Q$ and $J$ take
the following block-matrix structure:
\begin{equation}\label{vec1}
Q(x,t)=\left(\begin{array}{ccc}  0 & \vec{q}^{T} & 0 \\
  \vec{p} & 0 & s_{0}\vec{q} \\  0 & \vec{p}{\,}^T s_{0} & 0 \\
\end{array}\right),\qquad J=\mbox{diag}(1,0,\ldots 0, -1).
\end{equation}
For physical applications one uses mostly potentials satisfying the typical reduction, i.e.
$\vec{p}(x,t) = \vec{q}{\,}^*(x,t)$. The  vector $\vec{q}(x,t)$  for integer $F=r$ has $2r+1$ components
\begin{equation}\label{eq:sok}
\vec{q}(x,t) = (\Phi_{r-1}, \dots, \Phi_0 ,\dots,  \Phi_{-r+1})^T(x,t),
\end{equation}
and the corresponding matrices $s_0  $ enter in the definition of the orthogonal algebras $so(2r-1)$; namely
$X\in so(2r+1)$ if
\begin{equation}\label{eq:z1.6a}
X + S_0 X^T S_0 =0, \quad S_0  = \sum_{s=1}^{2r+1} (-1)^{s+1} E_{s, n+1-s}, \quad
S_0 = \left(\begin{array}{ccc}  0 & 0 & 1 \\ 0 & -s_0 & 0 \\  1 & 0 & 0 \\ \end{array}\right).
\end{equation}
By $E_{sp}$ above we mean $2r+1\times 2r+1$ matrix with matrix elements $(E_{sp})_{ij}=\delta_{si}\delta_{pj}$.
With the definition of orthogonality used in (\ref{eq:z1.6a}) the Cartan generators $H_k=E_{k,k}
-E _{2r+2-k,2r+2-k}$ are represented by diagonal matrices.

If we make use of the typical reduction $Q=Q^{\dag}$ (or $\vec{p}{\,}^* =\vec{q}$)
the generic MNLS type equations related to ${\bf BD.I.}$ acquire the form:
\begin{equation}\label{eq:4.2}
i \vec{q}_t+ \vec{q}_{xx} + 2 (\vec{q}{\,}^\dag,\vec{q}) \vec{q} -
(\vec{q},s_0\vec{q}) s_0\vec{q}{\,}^* =0.
\end{equation}
and for $r=2$ (resp. $r=3$) coincides with the MNLS eq. (\ref{eq:1}) (resp. with eq. (\ref{eq:3.1})).
The Hamiltonians for the MNLS equations (\ref{eq:4.2})  are given by
\begin{eqnarray}\label{eq:Ham1}
H_{{\rm MNLS}}=\int_{-\infty}^\infty d x
\left((\partial_{x}\vec{q}{\,}^\dag,\partial_{x}\vec{q})- (\vec{q}{\,}^\dag,\vec{q})^2+ \frac{1}{2}
| (\vec{q}^T,s_{0}\vec{q})|^2\right).
\end{eqnarray}

\subsection{The Direct and the Inverse scattering problem for $L$}\label{sec:3}

We remind some basic features of the scattering theory for the
Lax operators $L$, see  \cite{GKKV*08,GKV*09}. There we have made use of the general theory
developed in \cite{ZaSh1,ZaSh2,ZMNP,FaTa,VSG1} and the references therein.
The Jost solutions  of $L$ are defined by:
\begin{equation}
\lim_{x \to -\infty} \phi(x,t,\lambda) e^{  i \lambda J x
}=\openone, \qquad  \lim_{x \to \infty}\psi(x,t,\lambda) e^{  i
\lambda J x } = \openone
\end{equation}
and the scattering matrix $T(\lambda,t)\equiv \psi^{-1}\phi(x,t,\lambda)$. The special choice of $J$ and the
fact that the Jost solutions and the scattering matrix take values
in the group $SO(2r+1)$ we can use the following block-matrix structure of $T(\lambda,t)$
\begin{equation}\label{eq:25.1}
T(\lambda,t) = \left( \begin{array}{ccc} m_1^+ & -\vec{b}^-{}^T & c_1^- \\
\vec{b}^+ & {\bf T}_{22} & - s_0\vec{B}^- \\ c_1^+ & \vec{B}^+{}^Ts_0 & m_1^- \\
\end{array}\right), \qquad \hat{T}(\lambda,t) = \left( \begin{array}{ccc} m_1^- & \vec{B}^-{}^T & c_1^- \\
-\vec{B}^+ & {\bf \hat{T}}_{22} &  s_0\vec{b}^- \\ c_1^+ & -\vec{b}^+{}^Ts_0 & m_1^+ \\
\end{array}\right),
\end{equation}
where $\vec{b}^\pm (\lambda,t)$ and $\vec{B}^\pm (\lambda,t)$ are
$2r-1$-component vectors, ${\bf T}_{22}(\lambda)$ is $2r-1 \times 2r-1$ block matrix, and $m_1^\pm
(\lambda)$, and $c_1^\pm (\lambda)$ are scalar functions. Below we often use  $\hat{X}$ to denote
the matrix inverse to $X$.

\begin{remark}\label{rem:1}
The typical reduction $\vec{p}(x,t) = \vec{q}{\,}^*(x,t)$ mentioned above imposes on $T(\lambda,) $ the constraint
$T^\dag (\lambda ,t) = \hat{T}(\lambda ,t)$ for real values of $\lambda \in \bbbr$, i.e.
\begin{equation}\label{eq:r1}
\begin{aligned}
m_1^+(\lambda ) &= m_1^{-,*}(\lambda ), &\qquad \vec{B}_1{\,}^-(\lambda ) &= \vec{b}_1{\,}^{+,*}(\lambda ), \\
c_1^+(\lambda ) &= c_1^{-,*}(\lambda ), &\qquad \vec{B}_1{\,}^+(\lambda ) &= \vec{b}_1{\,}^{-,*}(\lambda ).
\end{aligned}
\end{equation}
\end{remark}

The Lax representation (\ref{eq:1}) allows one to prove that if $\vec{q}(x,t)$ satisfies the MNLS (\ref{eq:4.2})
then the scattering matrix $T(\lambda,t)$ satisfies the linear evolution equation \cite{GKV*09}:
\begin{equation}\label{eq:dTdt}
i\frac{dT}{dt} - \lambda^2 [J, T(\lambda,t)]=0,
\end{equation}
or in components:
\begin{equation}\label{eq:evol}
\begin{aligned}
i\frac{d\vec{b}^{\pm}}{d t} \pm \lambda ^2
\vec{b}^{\pm}(t,\lambda ) &=0, & \qquad i\frac{d\vec{B}^{\pm}}{d t}
\pm \lambda ^2 \vec{B}^{\pm}(t,\lambda ) &=0, \\
i\frac{d m_1^{\pm}}{d t}  &=0, &\qquad  i \frac{d{\bf
m}_2^{\pm}}{d t}  &=0.
\end{aligned}
\end{equation}
Thus the block-diagonal matrices $D^{\pm}(\lambda)$ can be considered as generating functionals of the
integrals of motion. Thus the problem of solving the MNLS eq. is based on the effective analysis
of the mapping between the potential $Q(x,t)$ of $L$ and the scattering matrix $T(\lambda,t)$.

\subsection{The fundamental analytic solution and the Riemann-Hilbert problem}

The most effective method for the above mentioned analysis consists in constructing the fundamental analytic solution
(FAS) of $L$-operators of type (\ref{eq:3.1}) and reducing the  inverse scattering problem to an
equivalent Riemann-Hilbert problem (RHP). Skipping the details (see \cite{GKV*09b})
we just outline the construction of FAS for $L$. Obviously the FAS, like any other
fundamental solutions of $L$ must be linearly related to the Jost solutions. For the class of potentials
$Q(x,t)$ with vanishing boundary conditions there exist two FAS $\chi^\pm(x,t,\lambda)$ which allow analytic
extension for $\lambda\in\bbbc_\pm$ respectively. For real $\lambda$ they are related to the Jost solutions by
\begin{equation}\label{eq:FAS_J}
\chi ^\pm(x,t,\lambda)= \phi (x,t,\lambda) S_{J}^{\pm}(t,\lambda )
= \psi (x,t,\lambda ) T_{J}^{\mp}(t,\lambda ) D_J^\pm (\lambda),
\end{equation}
where $T_{J}^{\mp}(t,\lambda )$, $ D_J^\pm (\lambda)$ and $T_{J}^{\mp}(t,\lambda )$ are the generalized Gauss factors
of $T(\lambda,t)$, see \cite{ZMNP,VSG2,TMF98}:
\begin{equation}\label{eq:25.1'}
\begin{aligned}
T(\lambda,t) &= T^-_J D^+_J \hat{S}^+_J , &\quad  T(\lambda,t) &= T^+_J D^-_J \hat{S}^-_J ,\\
T^\mp_J (\lambda,t) &= e^{\pm \left(\vec{\rho}^\pm , \vec{E}^\mp_1 \right)}  , &
S^\pm _J(\lambda,t)  &= e^{\pm \left(\vec{\tau}^\pm , \vec{E}^\pm_1 \right)} ,  & \\
D_J^\pm (\lambda)  &= \diag \left( (m_1^\pm)^{\pm 1} , {\bf m}_2^\pm , (m_1^\pm)^{\mp1} \right),
\end{aligned}
\end{equation}
Here
\begin{equation}\label{eq:25.1''}
\begin{aligned}
\vec{\tau}^\pm (\lambda,t) &= \left( \tau^\pm_{r-1},\dots, \tau^\pm_0, \dots, \tau^\pm_{-r+1} \right)^T (\lambda,t), \\
\left(\vec{\tau}^+ , \vec{E}^+_1 \right) &= \sum_{k=1}^{r-1} (\tau^+_{k} E_{e_1-e_{k+1}} + \tau^+_{-k}
E_{e_1+e_{k+1}}) + \tau^+_{0} E_{e_1}, \\  \left(\vec{\tau}^- , \vec{E}^-_1 \right) &=
\sum_{k=1}^{r-1} (\tau^-_{k} E_{-e_1+e_{k+1}} + \tau^-_{-k} E_{-e_1-e_{k+1}}) + \tau^-_{0} E_{-e_1},
\end{aligned}
\end{equation}
and similar expressions for $\left(\vec{\rho}^\pm , \vec{E}^\mp_1 \right)$.
Above we have made use of the fact that $\Delta_1^+$ consists of the roots
$\{ e_1-e_k, e_1, e_1+e_k\}_{k=1}^{r-1} $.
The functions $m_1^\pm$ and $n\times n$ matrix-valued functions ${\bf m}_2^\pm$
are analytic for $\lambda\in\bbbc_\pm$. One can check, that the analogs of the reflection coefficients
$\vec{\rho}^\pm$ and  $\vec{\tau}^\pm$ are expressed by:
\begin{equation*}\label{eq:25.1a}
\vec{\rho}^- =\frac{\vec{B}^-}{m_1^-},  \qquad \vec{\tau}^-
=\frac{\vec{B}^+}{m_1^-}, \qquad \vec{\rho}^+
=\frac{\vec{b}^+}{m_1^+},  \qquad \vec{\tau}^+
=\frac{\vec{b}^-}{m_1^+}.
\end{equation*}

\begin{remark}\label{rem:2}
The typical reduction means that for $\lambda \in \bbbr$ the reflection coefficients are
constrained by (see remark \ref{rem:1} above):
\begin{equation}\label{eq:r2}
\vec{\rho}{\,}^+(\lambda ) =\vec{\rho}{\,}^{-,*}(\lambda ) , \qquad
\vec{\tau}{\,}^+(\lambda ) =\vec{\tau}{\,}^{-,*}(\lambda ) , \quad \lambda \in \bbbr.
\end{equation}
\end{remark}

There are  some additional relations which ensure that
both $T(\lambda)$ and its inverse $\hat{T}(\lambda)$ belong to the
orthogonal group $SO(2r+1)$ and that $T(\lambda)\hat{T}(\lambda) =\openone$.

The FAS $\chi^\pm(x,t,\lambda)$ are related by:
\begin{equation}\label{eq:rhp0}
\chi^+(x,t,\lambda) =\chi^-(x,t,\lambda) G_{0,J}(\lambda,t),
\qquad G_{0,J}(\lambda,t)=\hat{S}^-_J(\lambda,t)S^+_J(\lambda,t)
\end{equation}

Below for convenience we introduce  $\xi^\pm(x,\lambda)=\chi^\pm(x,\lambda) e^{i\lambda Jx}$
which satisfy the equation:
\begin{equation}\label{eq:xi}
i\frac{d\xi^\pm}{dx} + Q(x)\xi^\pm(x,\lambda) -\lambda [J, \xi^\pm(x,\lambda)]=0,
\end{equation}
and  the relation
\begin{equation}\label{eq:rh-n}
\lim_{\lambda \to \infty} \xi^\pm(x,t,\lambda) = \openone, \qquad
\end{equation}
Then $\xi^\pm(x,\lambda)$ satisfy the RHP's
\begin{equation}\label{eq:rhp1}
\begin{split}
\xi^+(x,t,\lambda) &=\xi^-(x,t,\lambda) G_J(x,t,\lambda), \\
G_{J}(x,t,\lambda) &=e^{-i\lambda J(x+\lambda t)}G_{0,J}(\lambda)e^{i\lambda J(x+\lambda t)} ,
\end{split}
\end{equation}
with  sewing function $G_J(x,t,\lambda)$    uniquely determined by the Gauss factors $S_J^\pm (t,\lambda)$
taken for $t=0$:
\[ G_{0,J}(\lambda) = \hat{S}_J^-(0,\lambda) S_J^+(0,\lambda). \]
The analyticity properties of these FAS follow from the equivalent set of integral equations:
\begin{equation}\label{eq:rhpi}
\begin{aligned}
\xi^+_{1j}(x,\lambda ) &= \delta _{1j} + i \int_{\infty }^{x} dy e^{-i\lambda (x-y)} \sum_{p=1}^{2r-1} q_{r-p}(y)
\xi^+_{p+1,j}(y,\lambda ), & & \\
\xi^+_{kj}(x,\lambda ) &= \delta_{kj} + i \int_{-\infty }^{x} dy ( q_{r-k+1}^*(y) \xi^+_{1,j}(y,\lambda ) \\
& \qquad -(-1)^{r+k}  q_{-r+k-1}(y) \xi^+_{2r+1,j}(y,\lambda )), \qquad 2\leq k,j \leq 2r ; \\
\xi^+_{2r+1,j}(x,\lambda ) &= \delta _{2r+1,j} + i \int_{-\infty }^{x} dy e^{i\lambda (x-y)} \sum_{p=1}^{2r-1}
(-1)^{p+1} q_{-r+p}^*(y) \xi^+_{p+1,j}(y,\lambda),
\end{aligned}
\end{equation}
and a similar set of integral equations for
\begin{equation}\label{eq:rhp2}
\begin{aligned}
\xi^-_{1j}(x,\lambda ) &= \delta _{1j} + i \int_{-\infty }^{x} dy e^{-i\lambda (x-y)} \sum_{p=1}^{2r-1} q_{r-p}(y)
\xi^-_{p+1,j}(y,\lambda ), & & \\
\xi^-_{kj}(x,\lambda ) &= \delta_{kj} + i \int_{-\infty }^{x} dy ( q_{r-k+1}^*(y) \xi^-_{1,j}(y,\lambda )
\\ & \qquad -(-1)^{r+k}  q_{-r+k-1}(y) \xi^-_{2r+1,j}(y,\lambda )),  \qquad 2\leq k,j \leq 2r ; \\
\xi^-_{2r+1,j}(x,\lambda ) &= \delta _{2r+1,j} + i \int_{\infty }^{x} dy e^{i\lambda (x-y)} \sum_{p=1}^{2r-1}
(-1)^{p+1} q_{-r+p}^*(y) \xi^-_{p+1,j}(y,\lambda),
\end{aligned}
\end{equation}

The RHP (\ref{eq:rhp1}) with the additional condition (\ref{eq:rh-n}) is known as an RHP with
canonical normalization.

\begin{remark}\label{rem:1'}
An immediate consequence of the analyticity of $\xi^\pm(x,t,\lambda)$ is that $D^\pm(\lambda)$ are
analytic functions for $\lambda\in\bbbc_\pm$. This fact follows from the relation
$\lim_{x\to\infty}\xi^\pm(x,\lambda)= D^\pm(\lambda)$.
\end{remark}

Zakharov and Shabat proved a theorem \cite{ZaSh1,ZaSh2} which states that  if $G_J(x,\lambda,t)$  satisfies:
\begin{equation}\label{eq:G1}
\begin{split}
i\frac{dG}{dx} - \lambda [J,G(x,\lambda,t)] &=0, \\
i\frac{dG}{dt} - \lambda^2 [J,G(x,\lambda,t)] &=0,
\end{split}
\end{equation}
then the corresponding solutions of the RHP allow one to construct $\chi^\pm(x,\lambda)=\xi^\pm(x,\lambda) e^{-i\lambda Jx}$
as a fundamental solution of the Lax pair eq. (\ref{eq:1}).

We will say that $\xi^\pm_0(x,\lambda)$ is a regular solution to the RHP (\ref{eq:rhp1}) if
the block-diagonal part of it  has neither zeroes nor poles in its whole region of analyticity.

If we have solved the RHP's  and know the FAS $\xi^+(x,t,\lambda)$ then the formula
\begin{equation}\label{eq:XI-Q}
    Q(x,t) = \lim_{\lambda\to\infty} \lambda \left( J- \xi^+(x,t,\lambda)J\hat{\xi}^+(x,t,\lambda)
    \right),
\end{equation}
allows us to recover the corresponding potential of $L$.

\section{Singular solutions of RHP and soliton solutions of MNLS}

Zakharov-Shabat's theorem ensures that if a given  RHP allows regular solution,
then this solution is unique. However the RHP may have many singular solutions.
The construction of such singular solutions starting from a given regular one is
known as the dressing Zakharov-Shabat method \cite{ZaSh1,ZaSh2}. Indeed, if $\xi_0^\pm(x,t,\lambda)$
are regular solutions to the RHP, then
\begin{equation}\label{eq:uxi}
\xi^\pm(x,t,\lambda) = u(x,t,\lambda) \xi_0^\pm(x,t,\lambda)
\end{equation}
with conveniently chosen dressing factor $u(x,t,\lambda)$ may again be a solution of the RHP \cite{ZaSh1,ZaSh2}.
Obviously this factor must be analytic (with the exception of finite number of singular points) in the whole
complex $\lambda$-plane and can explicitly be constructed using only the solution of the regular RHP.

In order to obtain $N$-soliton solutions one has to apply the dressing procedure to the trivial
solution of the RHP $\xi_0(x,t,\lambda)=\openone$. We choose a dressing factor with  $2N$-poles \cite{GKV*09b}:
\begin{equation}
u(x,t,\lambda)=\openone+\sum^N_{k=1}\left(\frac{A_k(x,t)}{\lambda-\lambda^+_k}
+\frac{B_k(x,t)}{\lambda-\lambda^-_k}\right).	
\label{dressfac_bcd_n}\end{equation}
The $N$-soliton solution itself can be generated via the following formula
\begin{equation}
Q_{N,\rm s}(x,t)= \sum^N_{k=1}[J,A_k(x,t)+B_k(x,t)].	
\label{dressq_bcd_n}\end{equation}
The dressing
factor $u(x,\lambda ) $ must satisfy the equation
\begin{equation}\label{eq:u-eq}
i\frac{\partial u}{\partial x} + Q_{N,\rm s}(x,t) u(x,t,\lambda) -  \lambda [J,u(x,t,\lambda)] =0
\end{equation}
and the normalization condition $\lim_{\lambda \to\infty }
u(x,\lambda ) =\openone $.

The residues of $u$ admit the following decomposition
\[A_k(x,t)=X_k(x,t)F^T_k(x,t),\qquad B_k(x,t)=Y_k(x,t){G}^T_k(x,t), \]
where all matrices involved are supposed to be rectangular and of maximal rank $s$  \cite{ZaMi,GGK05a}.
By comparing the coefficients before the same powers of $\lambda-\lambda^{\pm}_k$
in (\ref{eq:u-eq}) we convince ourselves that the factors $F_k$ and $G_k$ can be
expressed by the fundamental analytic solutions $\chi^{\pm}_0(x,t,\lambda) = e^{-i\lambda(x+\lambda t)J}$ as follows
\[F^T_k(x,t)=F^T_{k,0}[\chi^{+}_0(x,t,\lambda^+_k)]^{-1},\qquad G^T_k(x,t)=G^T_{k,0}[\chi^{-}_0(x,t,\lambda^-_k)]^{-1}.\]
The constant rectangular matrices $F_{k,0}$ and $G_{k,0}$ obey the algebraic relations
\[F^T_{k,0}S_0F_{k,0}= 0,\qquad G^T_{k,0}S_0G_{k,0}=0.\]

The other two types of factors $X_k(x,t)$ and $Y_k(x,t)$ are solutions to the algebraic system
\begin{equation}\label{eq:XY}
\begin{split}
S_0F_k &=X_k\alpha_k+\sum_{l\neq k}\frac{X_lF^T_lS_0F_k}{\lambda^+_l-\lambda^+_k}
+\sum_l\frac{Y_lG^T_lS_0F_k}{\lambda^-_l-\lambda^+_k},\\
S_0G_k &=\sum_{l}\frac{X_lF^T_lS_0G_k}{\lambda^+_l-\lambda^-_k}+Y_k\beta_k
+\sum_{l\neq k}\frac{Y_lG^T_lS_0G_k}{\lambda^-_l-\lambda^-_k}.
\end{split}
\end{equation}
The square $s\times s$ matrices $\alpha_k(x,t)$ and $\beta_k(x,t)$ introduced above depend on
$\chi^+_0$ and $\chi^-_0$ and their derivatives by $\lambda$ as follows
\begin{equation}\label{eq:XY-alf}
\begin{split}
\alpha_k(x,t)&=-F^T_{0,k}[\chi^{+}_0(x,t,\lambda^{+}_k)]^{-1}
\partial_{\lambda}\chi^{+}_0(x,t,\lambda^{+}_k)S_0F_{0,k}+\alpha_{0,k},\\
\beta_k(x,t)&=-G^T_{0,k}[\chi^{-}_0(x,t,\lambda^{-}_k)]^{-1}
\partial_{\lambda}\chi^{-}_0(x,t,\lambda^{-}_k)S_0G_{0,k}+\beta_{0,k}.
\end{split}
\end{equation}

Below for simplicity we will choose $F_k$ and $G_k$ to be $2r+1$-component vectors. Then one can show 
that $\alpha_k=\beta_k=0$ which simplifies the system (\ref{eq:XY}). We also introduce the following
more convenient parametrization for $F_k$ and $G_k$, namely (see eq. (\ref{eq:njxt})):
\begin{equation}\label{eq:FkGk}
\begin{split}
F_k(x,t) &= S_0|n_k(x,t)\rangle = \left( \begin{array}{c} e^{-z_k+i\phi_k} \\
-\sqrt{2} s_0 \vec{\nu}_{0k} \\ e^{z_k-i\phi_k} \end{array} \right), \\
G_k(x,t) &= |n_k^*(x,t)\rangle = \left( \begin{array}{c} e^{z_k+i\phi_k} \\
\sqrt{2} \vec{\nu}{\,}^*_{0k} \\ e^{-z_k-i\phi_k} \end{array} \right),
\end{split}
\end{equation}
where $\vec{\nu}_{0k}$ are constant $2r-1$-component polarization vectors and
\begin{equation}\label{eq:njxt}
\begin{split}
z_j = \nu_j(x+2\mu_j t) + \xi_{0j} , &\qquad \phi_j = \mu_j x+(\mu_j^2-\nu_j^2) t +\delta_{0j}, \\
\langle n_j^T(x,t)| S_0|n_j(x,t)\rangle =0, \qquad &\mbox{or} \qquad (\vec{\nu}_{0,j} s_0 \vec{\nu}_{0,j}) = 1.
\end{split}
\end{equation}
With this notations the polarization vectors automatically satisfy the condition $\langle n_j (x,t) |S_0| n_j(x,t)\rangle =0$.
Thus for $N=1$ we get the system:
\begin{equation}\label{eq:N1}
|Y_1\rangle = -\frac{ (\lambda_1^+ -\lambda_1^-) |n_1\rangle}{\langle n_1^\dag | n_1\rangle} , \qquad
|X_1\rangle = \frac{ (\lambda_1^+ -\lambda_1^-) S_0 |n_1^*\rangle}{\langle n_1^\dag | n_1\rangle} ,
\end{equation}
which is easily solved. As a result for the one-soliton solution we get:
\begin{equation}\label{eq:1s}
\vec{q}_{\rm 1s} = -\frac{i \sqrt{2}(\lambda_1^+ -\lambda_1^-) e^{-i\phi_1} }{\Delta_1}
\left( e^{-z_1} s_0|\vec{\nu}_{01}\rangle + e^{z_1} |\vec{\nu}_{01}^*\rangle \right), \quad \Delta_1 =
\cosh (2z_1) +  \langle\vec{\nu}_{01}^\dag   |\vec{\nu}_{01}\rangle .
\end{equation}

For $n=3$ we put $\nu_{0k} = |\nu_{0k}|e^{\alpha_{0k}}$ and get:
\begin{equation}\label{eq:1s-1}
\begin{split}
\Phi_{{\rm 1s};\pm 1} &= -\frac{\sqrt{2 |\nu_{01;1}\nu_{01;3}|}(\lambda_1^+ -\lambda_1^-)  }{\Delta_1} e^{-i\phi_1 \pm i\beta_{13}} \\
& \times \left( \cosh(z_1 \mp \zeta_{01}) \cos(\alpha_{13}) - i\sinh(z_1 \mp \zeta_{01}) \sin(\alpha_{13})   \right), \\
\Phi_{{\rm 1s};0} &= -\frac{\sqrt{2} |\nu_{01;2}|(\lambda_1^+ -\lambda_1^-)  }{\Delta_1} e^{-i\phi_1 }
\left( \sinh{z_1 } \cos(\alpha_{02}) +i\cosh{z_1 } \sin(\alpha_{02})   \right), \\
\beta_{13} &= \frac{1}{2} (\alpha_{03}-\alpha_{01}) , \qquad \zeta_{01} = \frac{1}{2} \ln  \frac{|\nu_{01;3}|}{|\nu_{01;1}|},
\qquad \alpha_{13} = \frac{1}{2} (\alpha_{03}+\alpha_{01}) ,
\end{split}
\end{equation}
Note that the `center of mass` of $\Phi_{{\rm 1s};1}$  (resp. of $\Phi_{{\rm 1s};-1}$) is shifted with respect to the one of
$\Phi_{{\rm 1s};0}$ by $\zeta_{01}$ to the right (resp to the left); besides $|\Phi_{{\rm 1s};1}|=|\Phi_{{\rm 1s};-1}|$, i.e. they have
the same amplitudes.

For $n=5$ we put $\nu_{0k} = |\nu_{0k}|e^{\alpha_{0k}}$ and get analogously:
\begin{equation}\label{eq:1s-3}
\begin{split}
\Phi_{{\rm 1s};\pm 2} &= -\frac{\sqrt{2 |\nu_{01;1}\nu_{01;5}|}(\lambda_1^+ -\lambda_1^-)  }{\Delta_1} e^{-i\phi_1\pm i\beta_{15}} \\
& \times \left( \cosh(z_1\mp \zeta_{01}) \cos(\alpha_{15}) - i\sinh(z_1\mp \zeta_{01}) \sin(\alpha_{15})   \right), \\
\Phi_{{\rm 1s};\pm 1} &= \frac{\sqrt{2 |\nu_{01;2}\nu_{01;4}|}(\lambda_1^+ -\lambda_1^-)  }{\Delta_1} e^{-i\phi_1\pm i\beta_{24}}\\
& \times \left( \cosh(z_1\mp \zeta_{02}) \cos(\alpha_{24}) -i\sinh(z_1\mp \zeta_{01}) \sin(\alpha_{24})   \right),\\
\Phi_{{\rm 1s};0} &= -\frac{\sqrt{2} |\nu_{01;3}|(\lambda_1^+ -\lambda_1^-)  }{\Delta_1} e^{-i\phi_1 }
\left( \cosh{z_1 } \cos(\alpha_{03}) -i\sinh{z_1 } \sin(\alpha_{03})   \right), \\
\beta_{15} &= \frac{1}{2} (\alpha_{05}-\alpha_{01}) , \qquad \zeta_{01} = \frac{1}{2} \ln  \frac{|\nu_{01;5}|}{|\nu_{01;1}|},
\qquad \alpha_{15} = \frac{1}{2} (\alpha_{05}+\alpha_{01}) , \\
\beta_{24} &= \frac{1}{2} (\alpha_{04}-\alpha_{02}) , \qquad \zeta_{02} = \frac{1}{2} \ln  \frac{|\nu_{01;4}|}{|\nu_{01;2}|},
\qquad \alpha_{24} = \frac{1}{2} (\alpha_{04}+\alpha_{02}) .
\end{split}
\end{equation}
Similarly the `center of mass` of $\Phi_{{\rm 1s};2}$  and  $\Phi_{{\rm 1s};1}$  (resp. of $\Phi_{{\rm 1s};-2}$ and
$\Phi_{{\rm 1s};-1}$) are shifted with respect to the one of $\Phi_{{\rm 1s};0}$ by $\zeta_{01}$  and $\zeta_{02}$
to the right (resp to the left); besides $|\Phi_{{\rm 1s};2}|=|\Phi_{{\rm 1s};-2}|$ and $|\Phi_{{\rm 1s};1}|=|\Phi_{{\rm 1s};-1}|$.

For $N=2$ we get:
\begin{equation}\label{eq:n-nst}
\begin{split}
|n_1(x,t)\rangle & = \frac{X_2(x,t) f_{21}}{\lambda_2^+ -\lambda_1^+} +
\frac{Y_1(x,t) \kappa_{11}}{\lambda_1^- -\lambda_1^+} + \frac{Y_2(x,t) \kappa_{21}}{\lambda_2^- -\lambda_1^+}, \\
|n_2(x,t)\rangle & = \frac{X_1(x,t) f_{12}}{\lambda_1^+ -\lambda_2^+} +
\frac{Y_1(x,t) \kappa_{12}}{\lambda_1^- -\lambda_2^+} + \frac{Y_2(x,t) \kappa_{22}}{\lambda_2^- -\lambda_2^+}, \\
S_0|n_1^*(x,t)\rangle & = \frac{X_1(x,t) \kappa_{11}}{\lambda_2^+ -\lambda_1^+} +
\frac{X_2(x,t) \kappa_{11}}{\lambda_2^+ -\lambda_1^-} + \frac{Y_2(x,t) f_{21}^*}{\lambda_2^- -\lambda_1^-}, \\
S_0|n_2^*(x,t)\rangle & = \frac{X_1(x,t) \kappa_{21}}{\lambda_1^+ -\lambda_2^-} +
\frac{X_2(x,t) \kappa_{22}}{\lambda_2^+ -\lambda_2^-} + \frac{Y_1(x,t) f_{12}^*}{\lambda_1^- -\lambda_2^-},
\end{split}
\end{equation}
where
\begin{equation}\label{eq:n-nst2}
\begin{split}
\kappa_{kj}(x,t)& =  e^{z_k+z_j+i(\phi_k-\phi_j)} +e^{-z_k-z_j-i(\phi_k-\phi_j)}
+ 2 \left( \vec{\nu}{\,}^\dag_{0k},\vec{\nu}_{0j}\right) ,\\
f_{kj}(x,t) &= e^{z_k-z_j-i(\phi_k-\phi_j)} +e^{z_j-z_k+i(\phi_k-\phi_j)}
- 2 \left( \vec{\nu}_{0k}^T s_0\vec{\nu}_{0j}\right) ,
\end{split}
\end{equation}

In other words:
\begin{equation}\label{eq:MX-n}
\mathcal{M} \vec{X} \equiv \left( \begin{array}{cccc} 0 & \frac{f_{21}}{\lambda_2^+ -\lambda_1^+} &
\frac{\kappa_{11}}{\lambda_1^- -\lambda_1^+} & \frac{\kappa_{21}}{\lambda_2^- -\lambda_1^+} \\
\frac{f_{12}}{\lambda_1^+ -\lambda_2^+} & 0 & \frac{\kappa_{12}}{\lambda_1^- -\lambda_2^+} &
\frac{\kappa_{22}}{\lambda_2^- -\lambda_2^+} \\
\frac{\kappa_{11}}{\lambda_1^+ -\lambda_1^-} & \frac{\kappa_{12}}{\lambda_2^+ -\lambda_1^-} & 0
& \frac{f_{21}^*}{\lambda_2^- -\lambda_1^-}  \\
\frac{\kappa_{21}}{\lambda_1^+ -\lambda_2^-} & \frac{\kappa_{22}}{\lambda_2^+ -\lambda_2^-}
& \frac{f_{12}^*}{\lambda_1^- -\lambda_2^-} & 0 \\ \end{array} \right) \left( \begin{array}{c} X_1 \\ X_2 \\
Y_1 \\ Y_2 \\ \end{array} \right) = \left( \begin{array}{c} |n_1\rangle \\ |n_2\rangle  \\
S_0|n_1^*\rangle  \\ S_0|n_2^*\rangle \\ \end{array} \right) .
\end{equation}
We can rewrite $\mathcal{M}$ in block-matrix form:
\begin{equation}\label{eq:MM}
\begin{split}
\mathcal{M} &= \left( \begin{array}{cc} \mathcal{M}_{11} & \mathcal{M}_{12} \\ \mathcal{M}_{21} & \mathcal{M}_{22} \end{array}
\right) , \qquad \mathcal{M}_{22} = \mathcal{M}_{11}^*, \qquad  \mathcal{M}_{21} = -\mathcal{M}_{12}^T, \\
\mathcal{M}_{11} &=  \frac{f_{12}}{\lambda_2^+ -\lambda_1^+} \left( \begin{array}{cc} 0 & 1 \\ -1 & 0 \\ \end{array} \right),
\qquad   \mathcal{M}_{12} =    \left( \begin{array}{cc} \frac{\kappa_{11}}{\lambda_1^- -\lambda_1^+} &
\frac{\kappa_{21}}{\lambda_2^- -\lambda_1^+} \\ \frac{\kappa_{12}}{\lambda_1^- -\lambda_2^+} & \frac{\kappa_{22}}{\lambda_2^- -\lambda_2^+} \\ \end{array} \right).
\end{split}
\end{equation}
The inverse of $ \mathcal{M}  $ is given by:
\begin{equation}\label{eq:MM-1}
\begin{split}
\mathcal{M}^{-1} &= \left( \begin{array}{cc}
\mathcal{N}_1 ^{-1} & - \mathcal{N}_{1}^{-1}\mathcal{M}_{12} \hat{\mathcal{M}}_{11}^*\\
-\mathcal{N}_2^{-1}  \mathcal{M}_{21}  \hat{\mathcal{M}}_{11} & \mathcal{N}_2 ^{-1} \end{array}
\right) , \\
\mathcal{N}_1 &= \mathcal{M}_{11} - \mathcal{M}_{12}\hat{\mathcal{M}}_{11}^* \mathcal{M}_{21} , \qquad
\mathcal{N}_2 = \mathcal{M}_{11}^* - \mathcal{M}_{21}\hat{\mathcal{M}}_{11} \mathcal{M}_{12}
\end{split}
\end{equation}

From eqs. (\ref{eq:MX-n}) and (\ref{eq:MM-1}) we obtain \cite{GKV*09}:
\begin{equation}\label{eq:MM1'}
\begin{split}
|X_1\rangle &= \frac{1}{Z} \left( \frac{f_{12}^*}{\lambda_1^- -\lambda_2^-} |n_2\rangle -
\frac{\kappa_{22}}{\lambda_2^+ -\lambda_2^-} S_0 |n_1^*\rangle + \frac{\kappa_{12}}{\lambda_2^+
-\lambda_1^-} S_0 |n_2^*\rangle  \right), \\
|X_2\rangle &= \frac{1}{Z} \left( -\frac{f_{12}^*}{\lambda_1^- -\lambda_2^-} |n_1\rangle +
\frac{\kappa_{21}}{\lambda_1^+ -\lambda_2^-} S_0 |n_1^*\rangle - \frac{\kappa_{11}}{\lambda_1^+
-\lambda_1^-} S_0 |n_2^*\rangle  \right), \\
|Y_1\rangle &= \frac{1}{Z} \left( \frac{\kappa_{22}}{\lambda_2^+ -\lambda_2^-} |n_1\rangle -
\frac{\kappa_{21}}{\lambda_1^+ -\lambda_2^-} |n_2\rangle - \frac{f_{12}}{\lambda_1^+
-\lambda_2^+} S_0 |n_2^*\rangle  \right), \\
|Y_2\rangle &= \frac{1}{Z} \left( -\frac{\kappa_{12}}{\lambda_2^+ -\lambda_1^-} |n_1\rangle +
\frac{\kappa_{11}}{\lambda_1^+ -\lambda_1^-} |n_2\rangle + \frac{f_{12}}{\lambda_2^+
-\lambda_1^+} S_0 |n_1^*\rangle  \right),
\end{split}
\end{equation}
where
\begin{equation}\label{eq:Z}
Z=  \left( \frac{|f_{12}|^2}{|\lambda_2^+-\lambda_1^+|^2} -
\frac{\kappa_{12}\kappa_{21}}{|\lambda_2^+-\lambda_1^-|^2}  +\frac{\kappa_{11}\kappa_{22}}{4\nu_1\nu_2} \right).
\end{equation}

Inserting this result into eq. (\ref{dressq_bcd_n}) we obtain the following expression for the 2-soliton
solution of the MNLS:
\begin{equation}\label{eq:2-sol}
\begin{split}
Q_{\rm 2s}(x,t) &= [J, A_1+B_1+A_2+B_2] = \frac{1}{Z} [J , C(x,t) - S_0C^T(x,t) S_0], \\
C(x,t) &=  \frac{\kappa_{22}}{\lambda_2^+ -\lambda_2^-} |n_1\rangle  \langle n_1^\dag | - \frac{\kappa_{12}}{\lambda_2^+ -\lambda_1^-}
|n_1\rangle \langle n_2^\dag | - \frac{\kappa_{21}}{\lambda_1^+ -\lambda_2^-} |n_2\rangle  \langle n_1^\dag |
\\ & + \frac{\kappa_{11}}{\lambda_1^+ -\lambda_1^-} |n_2\rangle  \langle n_2^\dag |   -
\frac{f_{12}^*}{\lambda_1^- -\lambda_2^-} |n_1\rangle \langle n_2| S_0 - \frac{f_{12}}{\lambda_1^+ -\lambda_2^+}
S_0|n_2^*\rangle  \langle n_1^\dag | .
\end{split}
\end{equation}

\section{The minimal sets of scattering data}

It is well known that the locations of the singularities of the RHP
$ \lambda_k^\pm \in \bbbc_\pm$ are zeroes of the functions $m_1^\pm(\lambda)$
and discrete eigenvalues of the Lax operator $L$. We will say that these
eigenvalues are simple if the corresponding eigensubspaces are one dimensional.
This corresponds to our choice of $F_k(x,t)$ and $G_k(x,t)$ as vectors.
Eigensubspaces of higher multiplicities $s>1$ can be obtained choosing  $F_k(x,t)$ and $G_k(x,t)$
as $s\times (2r+1)$ matrices of rank $s$.

\begin{theorem}\label{lem:ms}
Let the potential $Q(x,t)$ be such that the corresponding Lax operator $L$
has finite number of simple discrete eigenvalues located at the points
$ \lambda_k^\pm \in \bbbc_\pm$ respectively, $k=1,\dots, N$. Then as minimal sets of
scattering data  uniquely  determining both the scattering matrix $T(\lambda,t)$ and
the corresponding potential $Q(x,t)$ one can consider the sets
\begin{equation}\label{eq:T_i}
\begin{split}
& \mathfrak{T}_1 \equiv \{\vec{\tau}{\,}^{+}(\lambda,0), \quad \lambda \in \bbbr ;
\qquad \vec{\tau}_k^+ ,\;\lambda_k^+ \in \bbbc_+,\; k=1,\dots N \},
\\
&\mathfrak{T}_2 \equiv \{ \vec{\rho}{\,}^{\pm}(\lambda,0),\quad   \lambda
\in \bbbr; \qquad \vec{\rho}_k^+ , \; \lambda_k^+\in \bbbc_+ , \; k=1,\dots N\},
\end{split}\end{equation}
where the constant vectors $\vec{\tau}_{0k}^+ $ and $\vec{\rho}_{0k}^\pm $
\begin{equation}\label{eq:nknk1}
\begin{split}
\vec{\tau}{\,}^+_{0k} =  \left( \begin{array}{c} e^{\xi_{0k}-i\delta _{0k}} \\
\sqrt{2} \vec{\nu}_{0k} \\ e^{-\xi_{0k}+i\delta _{0k}} \end{array} \right), \qquad
\vec{\rho}{\,}^+_{0k} =  \left( \begin{array}{c} e^{\eta_{0k}-i\theta _{0k}} \\
\sqrt{2} \vec{\mu}_{0k} \\ e^{-\eta_{0k}+i\theta _{0k}} \end{array} \right),
\end{split}
\end{equation}
and the vectors $\vec{\nu}_{0k}$ and $\vec{\mu}_{0k}$ satisfy the normalization
condition  $(\vec{\nu}_{0k}^T s_0\vec{\nu}_{0k})=1$ and $(\vec{\mu}_{0k}^T s_0\vec{\mu}_{0k})=1$.
\end{theorem}

\begin{remark}\label{rem:5}
The data $\lambda _k^+$ and $\lambda _k^-=(\lambda _k^+)^*$ characterize the discrete eigenvalues of $L$.
The vectors $\vec{\tau}{\,}^+_{0k}$ and $\vec{\tau}{\,}^-_{0k}=(\vec{\tau}{\,}^+_{0k})^*$
(resp. $\vec{\rho}{\,}^+_{0k}$ and $\vec{\rho}{\,}^-_{0k}=(\vec{\rho}{\,}^+_{0k})^*$)
determine the corresponding eigenfunction of $L$. Note also that by definition these vectors
satisfy $(\vec{\tau}{\,}^{+,T}_{0k}s_0 \vec{\tau}{\,}^+_{0k})=0$ and
$(\vec{\rho}{\,}^{+,T}_{0k}s_0 \vec{\rho}{\,}^+_{0k})=0$.
\end{remark}

\begin{proof}[Outline of the proof.] Let us be given $\mathcal{T}_1 $. Using $\vec{\tau}{\,}^+(\lambda ,t)$ and
$\vec{\tau}{\,}^-(\lambda ,t)=(\vec{\tau}{\,}^+(\lambda ,t))^*$ we construct $S_{0J}^+(\lambda ,t)$ and
 $S_{0J}^-(\lambda ,t)$ and therefore obtain also the sewing function  $G_{0}(\lambda ,t) =\hat{S}_{0J}^-(\lambda ,t)
S_{0J}^+(\lambda ,t)$ for a regular RHP. According to the Zakharov-Shabat theorem it has unique solution
$\xi^\pm_0 (x,t,\lambda )$. The corresponding regular potential is obtained by:
\begin{equation}\label{eq:Q0}
\begin{split}
Q_0(x,t) &= \lim_{\lambda \to\infty} \lambda \left(J - \xi^\pm_0 (x,t,\lambda )J\hat{\xi}^\pm_0 (x,t,\lambda )\right) \\
&= [J, \xi_{01}^+(x,t) ],
\end{split}
\end{equation}
where $\xi_{01}^+(x,t) =\lim_{\lambda \to \infty} \lambda (\xi^+_0 (x,t,\lambda ) -\openone )$.

Next we use the dressing method to dress the regular solution $\xi^\pm_0 (x,t,\lambda )$ with the dressing
factor $u(x,t,\lambda )$ of the form (\ref{dressfac_bcd_n}). In order to do it we make use of the
set of eigenvalues $\lambda _k^+$ and $\lambda _k^- =(\lambda _k^+)^*$ and instead  of the polarization vectors
(\ref{eq:FkGk}) we use:
\begin{equation}\label{eq:nnk}
\begin{split}
|n_k(x,t)\rangle & = \xi_0^+(x,t,\lambda _k^+)e^{-i\lambda _k^+ (x +\lambda _k^+t)J} \vec{\tau}_{0k}{}^+ , \\
|n_k^*(x,t)\rangle & = \xi_0^-(x,t,\lambda _k^-) e^{-i\lambda _k^- (x +\lambda _k^-t)J}\vec{\tau}_{0k}{}^- .
\end{split}
\end{equation}
After solving the algebraic equations for $|X_k(x,t)\rangle$ and $|Y_k(x,t)\rangle$ we find explicitly
the dressed potential
\begin{equation}\label{eq:Qdr}
Q(x,t) = Q_0(x,t) + \sum_{k=1}^{N} [J, A_k(x,t) + B_k(x,t)],
\end{equation}
which proves the first part of the theorem.

Let us now show how one can recover $T(\lambda ,t)$ from $\mathcal{T}_1 $. Given the regular solution
$\xi_0^\pm(x,t,\lambda )$ we can find
\begin{equation}\label{eq:Dpm}
D_{0,J}^\pm (\lambda ) =\lim_{x\to\infty} \xi_0^\pm(x,t,\lambda ),
\end{equation}
and also
\begin{equation}\label{eq:DT}
T_{0,J}^\mp (\lambda ) D_{0,J}^\pm (\lambda ) =\lim_{x\to\infty} e^{i (\lambda x + \lambda ^2t)J }
\xi_0^\pm(x,t,\lambda ) e^{-i (\lambda x + \lambda ^2t)J }.
\end{equation}
Thus we have recovered all Gauss factors $T_{0,J}^\mp (\lambda )$, $ D_{0,J}^\pm (\lambda )$ and
$S_{0,J}^\pm (\lambda )$ of the `undressed' scattering matrix $T_0(\lambda ,t)$, so
\begin{equation}\label{eq:T0}
T_{0,J} (\lambda ,t) =  T_{0,J}^\mp (\lambda,t )  D_{0,J}^\pm (\lambda ) \hat{S}_{0,J}^\pm (\lambda ,t).
\end{equation}
In order to take into account the effect of dressing we make use of the relations between the dressed
and undressed Jost solutions:
\begin{equation}\label{eq:Jsol0}
\begin{aligned}
\psi(x,t,\lambda ) &= u(x,t,\lambda ) \psi_0(x,t,\lambda )\hat{u}_+(\lambda ), \\
\phi(x,t,\lambda ) &= u(x,t,\lambda ) \phi_0(x,t,\lambda )\hat{u}_-(\lambda ),
\end{aligned}
\end{equation}
where $u_\pm(\lambda )= \lim_{x\to\pm \infty} u(x,t,\lambda )$. As a result we get:
\begin{equation}\label{eq:T0T}
\begin{split}
T(\lambda ,t) &= \hat{\psi}(x,t,\lambda ) \phi(x,t,\lambda ) \\
&= u_+(\lambda )\hat{\psi}_0(x,t,\lambda ) \phi_0(x,t,\lambda ) \hat{u}_-(\lambda ) \\
&= u_+(\lambda )T_0(\lambda ,t)\hat{u}_-(\lambda ).
\end{split}
\end{equation}
Skipping the details we state the result:
\begin{equation}\label{eq:upm}
u_+(\lambda ) = \left(\begin{array}{ccc} c(\lambda ) & 0 & 0 \\ 0 & \openone & 0 \\ 0 & 0 & 1/c(\lambda )\\
\end{array}\right) , \qquad
u_-(\lambda ) = \left(\begin{array}{ccc} 1/c(\lambda ) & 0 & 0 \\ 0 & \openone & 0 \\ 0 & 0 & c(\lambda )\\
\end{array}\right),
\end{equation}
where $c(\lambda ) = \prod_{j=1}^{N} \frac{\lambda - \lambda _j^+}{\lambda - \lambda _j^-}$.

The fact that the set $\mathcal{T}_2 $ is also a minimal set of scattering data is proved analogously.

\end{proof}

\section{Reductions of MNLS}

Along with the typical reduction $Q=Q^\dag$ mentioned above one can impose additional reductions
using the reduction group proposed by Mikhailov \cite{Mikh}.
They are automatically compatible with the Lax representation of the corresponding MNLS eq. Below we
make use of two types of $\bbbz_2$-reductions\cite{R95}:
\begin{equation}\label{eq:U-V.a}
\begin{aligned}
&\mbox{1)} &\quad C_1 U^{\dagger}(x,t, \lambda^* ) C_1^{-1}&= U(x,t,\lambda ),
&\quad C_1 V^{\dagger}( x,t,\lambda^* ) C_1^{-1}&= V(x,t,\lambda ), \\
&\mbox{2)}  &\quad C_2 U^{T}(x,t,\lambda )C_2^{-1} &= -U(x,t,\lambda ), &\quad
C_2 V^{T}( x,t,\lambda )C_2^{-1} &= -V(x,t,\lambda ),
\end{aligned}
\end{equation}
where $C_1$ and $C_2$ are involutions of the Lie algebra $so(2r+1$, i.e.   $C_i^2=\openone$. They
can be chosen to be either diagonal (i.e., elements of the Cartan subgroup of $SO(2r+1)$) or
elements of the Weyl group.

The typical reductions of the MNLS eqs. is a  class 1) reduction obtained by specifying
$C_1$ to be the identity automorphism of $\fr{g}$; below we list several choices for $C_1$
leading to inequivalent reductions:
\begin{equation}\label{eq:C1-1}
\begin{aligned}
\mbox{1a)} \quad C_1 &=\openone, &\quad \vec{p}(x) &=\vec{q}{\,}^*(x) , \quad
&\mbox{1b)} \quad C_1&=K_1, & \vec{p}(x) &=K_{01}\vec{q}{\,}^*(x) , \\
\mbox{1c)} \quad C_1 &=S_{e_2}, & \quad \vec{p}(x) &=K_{02}\vec{q}{\,}^*(x) , \quad &\mbox{1d)}
\quad C_1&=S_{e_2}S_{e_3},  & \vec{p}(x) &=K_{03}\vec{q}{\,}^*(x) .
\end{aligned}
\end{equation}
We also make use of type 2) reductions:
\begin{equation}\label{eq:C1-2}
\begin{aligned}
\mbox{2e)} \quad C_2 &=K_4, &\quad \vec{q}(x) &=- K_{04} s_0\vec{q}(x) , \quad
&  \quad   \vec{p}(x) &=-K_{04} s_0\vec{p}(x) ,\\
\mbox{2f)} \quad C_2 &=K_5, &\quad \vec{q}(x) &= K_{05} \vec{q}(x) , \quad
&  \quad   \vec{p}(x) &=K_{05} \vec{p}(x) ,
\end{aligned}
\end{equation}
where
\begin{equation}\label{eq:C1-2'}
K_j =\bdiag(1, K_{0j}, 1), \qquad K_{01} = \diag(\epsilon_1,\dots ,\epsilon_{r-1},1, \epsilon_{r-1}, \dots,
\epsilon_1 ),
\end{equation}
for $j=1,2,3,5$ and $\epsilon_j=\pm 1$.
The matrices $K_{02}$, $K_{03}$ and $K_4$ are not diagonal and may take the form:
\begin{equation}\label{eq:C1-3}
\begin{aligned}
 K_{02} &= \left( \begin{array}{ccc} 0 & 0 & 1 \\ 0 & -1 & 0 \\ 1 & 0 & 0 \end{array}\right), &
 \quad  K_{4} &= \left( \begin{array}{ccc} 0 & 0 & 1 \\ 0 & K_{04} & 0 \\ 1 & 0 & 0 \end{array}\right), & \\
 K_{02} &= \left( \begin{array}{ccccc} 0 & 0 & 0 & 0 & -1 \\ 0 & 1 &0 & 0 & 0 \\
0 & 0 & -1 & 0 & 0 \\ 0 & 0 & 0 & 1 & 0  \\ -1 & 0 & 0 & 0 & 0  \end{array}\right), & \quad
K_{03} &= \left( \begin{array}{ccccc} 0 & 0 & 0 & 0 & -1 \\ 0 & 0 &0 & 1 & 0 \\
0 & 0 & -1 & 0 & 0 \\ 0 & 1 & 0 & 0 & 0  \\ -1 & 0 & 0 & 0 & 0  \end{array}\right). &
\end{aligned}
\end{equation}

Each of the above reductions impose constraints on the FAS, on the scattering matrix $T(\lambda)$ and on its
Gauss factors $S^\pm_J(\lambda)$, $T^\pm_J(\lambda)$ and $D^\pm_J(\lambda)$. For the type 1 reductions (cases 1a) -- 1d))
these have the form:
\begin{equation}\label{eq:C1-3'}
\begin{aligned}
(S^+(\lambda^*))^\dag &= K_j^{-1}\hat{S}^-(\lambda)K_j  &\quad (T^+(\lambda^*))^\dag &=K_j^{-1}\hat{T}^-(\lambda)K_j  &\\
(D^+(\lambda^*))^\dag &= K_j^{-1}\hat{D}^-(\lambda)K_j &&  \\
\vec{\tau}{\,}^+ &= K_{0j}\vec{\tau}{\,}^{-,*}, &\quad  \vec{\rho}{\,}^+ &= K_{0j}\vec{\rho}{\,}^{-,*}, & j=1,2,3
\end{aligned}
\end{equation}
where the matrices $K_j$ are specific for each choice of the automorphisms $C_1$, see eqs. (\ref{eq:C1-1}).
In particular, from the last line of (\ref{eq:C1-3'}) and (\ref{eq:C1-2}) we get:
\begin{equation}\label{eq:m1pm}
(m_1^+(\lambda^*))^* = m_1^-(\lambda),
\end{equation}
and consequently, if $m_1^+(\lambda)$ has zeroes at the points $\lambda_k^+$, then
$m_1^-(\lambda)$ has  zeroes at:
\begin{equation}\label{eq:lapm}
\lambda_k^- = (\lambda_k^+)^*, \qquad k=1,\dots, N.
\end{equation}

For the type 2) reductions we obtain:
\begin{equation}\label{eq:C1-3t1}
\begin{aligned}
& \mbox{2e)} & (S^\pm(\lambda))^T &= K_4^{-1}\hat{S}^\pm (\lambda)K_4  &\quad (T^\pm(\lambda))^T &=K_4^{-1}\hat{T}^\pm(\lambda)K_4 \\
&  &  (D^\pm(\lambda))^T &= K_4^{-1}\hat{D}^\pm (\lambda)K_4  && \\
& & \vec{\tau}{\,}^\pm &= -K_{04} s_0\vec{\tau}{\,}^{\pm}, &\quad  \vec{\rho}{\,}^\pm &= -K_{04}s_0\vec{\rho}{\,}^{\pm},
\end{aligned}
\end{equation}
and
\begin{equation}\label{eq:C1-3t2}
\begin{aligned}
& \mbox{2f)} &\quad (S^+(\lambda))^T &= K_5^{-1}\hat{S}^-(-\lambda)K_5  &\quad (T^+(\lambda))^T &=K_5^{-1}\hat{T}^-(-\lambda)K_5 \\
& & (D^+(\lambda))^T &= K_5^{-1}\hat{D}^-(-\lambda)K_5  && \\
& & \vec{\tau}{\,}^+(\lambda) &= K_{05} \vec{\tau}{\,}^{-} (-\lambda), &\quad  \vec{\rho}{\,}^+(\lambda) &=
-K_{05}\vec{\rho}{\,}^{-}(-\lambda),
\end{aligned}
\end{equation}
For the 2e) reduction with $n=3$ we may choose $K_4$ to corresponds to the Weyl group
element $S_{e_1}$, so $K_{04}=\openone$. As a result we get:
\begin{equation}\label{eq:red1}
\Phi_1= -\Phi_{-1}
\end{equation}
and $\Phi_0$  arbitrary.
This reduction of eq. (\ref{eq:1}) is also important for the BEC \cite{Kevre*08}. From
(\ref{eq:C1-3t1}) we find $\nu_{01}=\nu_{03}$. The effect of this constraint is that for the one-soliton
solution we get $\Phi_{{\rm 1s};1} =-\Phi_{{\rm 1s};-1}$.

Our next remark following  \cite{ps99} is that this reduction applied to the $F=1$ MNLS (\ref{eq:1}) leads to a 2-component MNLS
which after the change of variables
\begin{equation}\label{eq:2MNLS'}
\Phi_1 = \frac{1}{2} (w_1 + iw_2), \qquad \Phi_0  =    \frac{i}{\sqrt{2}}(w_1 - iw_2),
\end{equation}
leads to two disjoint NLS equations for $w_1$ and $w_2$ respectively.

It is only logical that applying the constraint $\nu_{01}= \nu_{03}$ the explicit expression for the one-soliton
solution (\ref{eq:1s-1}) simplifies and reduces to the standard soliton solutions of the scalar NLS.

For the other two examples of type 2) reductions we choose $n=5$ and $K_4$, and $K_5$  correspond to the Weyl group elements $S_{e_2}S_{e_3}$ and
$S_{e_2-e_3}$ respectively. Then $K_{04} = -s_0$ and
\begin{equation}\label{eq:k05}
K_{05} = \left( \begin{array}{ccccc} 0 & 1 & 0 & 0 & 0 \\ 1 & 0 &0 & 0 & 0 \\
0 & 0 & 1 & 0 & 0 \\ 0 & 0 & 0 & 0 & -1  \\ 0 & 0 & 0 & -1 & 0  \end{array}\right),
\end{equation}
For these choices of $K_4$, $K_5$ we obtain:
\begin{equation}\label{eq:red2}
\begin{aligned}
& \mbox{2e)} &\quad \Phi_{ 2} &=   \Phi_{-2}, &\qquad \Phi_{1} &=   \Phi_{- 1},\\
& \mbox{2f)} &\quad
\Phi_{\pm 2} &= \pm \frac{c}{\sqrt{1+c^2}} \Phi_{\pm 1}', &\qquad \Phi_{\pm 1} &= \frac{1}{\sqrt{1+c^2}} \Phi_{\pm 1}',
\end{aligned}
\end{equation}
It reduces the $F=2$ spin BEC model into the $F=1$ model.

The corresponding relations for the Gauss factors and for the polarization vectors are given by:

\begin{equation}\label{eq:red3}
\Phi_{\pm 2} = \pm \frac{c}{\sqrt{1+c^2}} \Phi_{\pm 1}', \qquad \Phi_{\pm 1} = \frac{1}{\sqrt{1+c^2}} \Phi_{\pm 1}',
\end{equation}

\section{Two Soliton interactions}

In this section we generalize the classical results of Zakharov and Shabat about soliton interactions \cite{zs72}  to the class
of MNLS equations related to BD.I symmetric spaces. For detailed exposition see the monographs \cite{ZMNP,FaTa}. These results were generalized
for the vector nonlinear Schr\"odinger equation by Manakov \cite{ma74}, see also \cite{AblPrinTru*04,Laksh,Tsu}.
The Zakharov Shabat approach consisted in calculating the asymptotics of generic $N$-soliton solution of NLS for $t \to \pm\infty$ and establishing
the pure elastic character of the generic soliton interactions. By generic here we mean $N$-soliton solution whose parameters
$\lambda_k^\pm = \mu_k \pm i\nu_k$ are such that $\mu_k \neq \mu_j$ for $k\neq j$. The pure elastic character of the soliton interactions
is demonstrated by the fact that for $t\to\pm\infty$ the generic $N$-soliton solution splits into sum of $N$ one soliton solutions each
preserving  its amplitude $2\nu_k$ and velocity $\mu_k$. The only effect of the interaction consists in shifting the center of mass
and the initial phase of the solitons. These shifts can be expressed in terms of $\lambda_k^\pm$ only; for detailed exposition see \cite{FaTa}.

Let us apply these ideas to the MNLS equations studied above.
Namely we  use the $2$-soliton solution (\ref{eq:2-sol}) derived above and
calculate its asymptotics   along the trajectory of the first soliton. To this end
we keep $z_1(x,t)$ fixed and let $\tau =z_2 -z_1$ tend to $\pm \infty$. Therefore it will  be enough to insert the asymptotic values of
the matrix elements of $\mathcal{M}$ for $\tau\to\pm\infty$ and keep only the leading terms. For $\tau\to\infty$ that gives:
\begin{equation}\label{eq:taup}
\begin{aligned}
\kappa_{22} & \simeq e^{2\tau} \exp(\nu_2z_1/\nu_1) + 2 \mathcal{C}_1 ,\\
\kappa_{12} &=   e^{\tau} \exp((1+\nu_2/\nu_1)z_1 +i(\phi_1 -\phi_2)) + \mathcal{O}(1) ,\\
\kappa_{21} &=  e^{\tau} \exp((1+\nu_2/\nu_1)z_1 -i(\phi_1 -\phi_2)) + \mathcal{O}(1) ,\\
f_{12} &=  e^{\tau} \exp(-(1-\nu_2/\nu_1)z_1 +i(\phi_1 -\phi_2)) + \mathcal{O}(1) ,
\end{aligned}
\end{equation}
while for $\tau\to -\infty$ we get:
\begin{equation}\label{eq:taum}
\begin{aligned}
\kappa_{22} & \simeq e^{-2\tau} \exp(-\nu_2z_1/\nu_1) + 2 \mathcal{C}_1 ,\\
\kappa_{12} &=   e^{-\tau} \exp(-(1+\nu_2/\nu_1)z_1 -i(\phi_1 -\phi_2)) + \mathcal{O}(1) ,\\
\kappa_{21} &=  e^{-\tau} \exp(-(1+\nu_2/\nu_1)z_1 +i(\phi_1 -\phi_2)) + \mathcal{O}(1) ,\\
f_{12} &=  e^{-\tau} \exp((1-\nu_2/\nu_1)z_1 -i(\phi_1 -\phi_2)) + \mathcal{O}(1) ,
\end{aligned}
\end{equation}
After somewhat lengthy  calculations we get:
\begin{equation}\label{eq:Zpm}
\begin{split}
\lim_{\tau\to\infty} \vec{q}_{\rm 2s}(x,t) &=-\frac{i\sqrt{2} \nu_1 e^{-i(\phi_1 -\alpha_+)} \left( e^{-z_1-r_+} s_0|\vec{\nu}_{01}\rangle
+  e^{z_1+r_+} |\vec{\nu}_{01}^*\rangle \right) }{\cosh(2(z_1+ r_+)) + (\vec{\nu}_{01}^\dag,\vec{\nu}_{01})},  \\
\lim_{\tau\to -\infty} \vec{q}_{\rm 2s}(x,t) &=\frac{i\sqrt{2} \nu_1 e^{-i(\phi_1 +\alpha_+)} \left( e^{-z_1+r_+} s_0|\vec{\nu}_{01}\rangle
+  e^{z_1-r_+} |\vec{\nu}_{01}^*\rangle \right) }{\cosh(2(z_1- r_+)) + (\vec{\nu}_{01}^\dag,\vec{\nu}_{01})},
\end{split}
\end{equation}
where
\[ r_+ =  \ln \left| \frac{\lambda_1^+ -\lambda_2^+}{\lambda_1^+ -\lambda_2^-} \right|, \qquad
\alpha_+ = \arg \frac{\lambda_1^+ -\lambda_2^+}{\lambda_1^+ -\lambda_2^-}.  \]

For $n=3$ and $n=5$ the right hand sides of (\ref{eq:Zpm}) coincide with the one-soliton solutions (\ref{eq:1s-1}) and
(\ref{eq:1s-3}) respectively.
This means that the $2$-soliton interaction for the above MNLS eqs. is purely elastic. The solitons preserve their shapes
and velocities and the only effect of the interaction consist in shifts of the center of mass and the phase. From this
point of view  the interaction is the same like for the scalar NLS eq.

It is important to check whether the $N$-soliton interactions consist of sequence of elementary 2-soliton
interactions and the shifts are additive.

\section{Effects of reductions and initial conditions on MNLS}

\begin{theorem}\label{th:1}
Let the minimal set of scattering data $\mathcal{T}_j$, $j=1,2$ for $t=0$ satisfy
the reduction conditions  (\ref{eq:C1-3t1}). Then
the solution $\vec{q}(x,t)$ of the MNLS with such initial data will satisfy the corresponding reduction 2e)
(\ref{eq:C1-2}).

\end{theorem}

\begin{proof}
Let the minimal sets of scattering data, say $\mathcal{T}_1$ satisfy the reduction conditions (\ref{eq:C1-3t1})
for $t=0$.  It is easy to check that  their evolution law (\ref{eq:evol}) is compatible with the
reduction, so (\ref{eq:C1-3t1}) will hold for all $t>0$. As a result the corresponding Gauss factors $S^\pm$,
$T^\pm$ and $D^\pm$, and consequently, the sewing function in the RHP $G(x,t,\lambda)$ will
satisfy
\begin{equation}\label{eq:Gred}
G(x,t,\lambda ) = K_4^{-1} \hat{G}^T(x,t,\lambda )K_4.
\end{equation}
The next consequence is that both $\xi^\pm $ and $K_4^{-1}\hat{\xi}^{\pm,T} K_4 $
are solutions of the RHP (\ref{eq:rhp1}) with the same sewing function and the same
canonical normalization. Therefore from the uniqueness of the solution of RHP
we get that the regular solutions of this RHP satisfy:
\begin{equation}\label{eq:xired}
\xi_0^\pm(x,t,\lambda ) = K_4^{-1} \hat{\xi}_0^{\pm, T}(x,t,\lambda )K_4.
\end{equation}
Next we note that the scattering data related to the discrete spectrum also satisfy the
reduction conditions. This means that the dressing factor $u(x,t,\lambda )$ and the singular
solutions $\xi^\pm(x,t,\lambda )=u(x,t,\lambda ) \xi_0^\pm(x,t,\lambda )\hat{u}_-(\lambda )$
also satisfy:
\begin{equation}\label{eq:xired1}
u(x,t,\lambda ) = K_4^{-1} \hat{u}^{T}(x,t,\lambda )K_4, \qquad
\xi^\pm(x,t,\lambda ) = K_4^{-1} \hat{\xi}^{\pm, T}(x,t,\lambda )K_4.
\end{equation}
It remains to check that from equations (\ref{eq:XI-Q}) and (\ref{eq:xired1}) there follows:
\begin{equation}\label{eq:Qred}
Q(x,t) = -K_4^{-1} Q^T(x,t) K_4.
\end{equation}
\end{proof}

\begin{remark}\label{rem:6}
Note that the above arguments are not specific for the choice of $K_4$.
The above theorem can be proved along the same lines for any reduction of type 1 and
type 2.
\end{remark}

A simple consequence of the above theorem is the following.
Consider $n=3$ and choose $\tau_1^+= \tau_3^+$ for $t=0$. Then the corresponding solution of $F=1$ BEC (\ref{eq:1})
will also satisfy $\Phi_1=-\Phi_{-1}$ for all $t>0$, i.e. will be a solution to
\begin{equation}\label{eq:fi1}
\begin{aligned}
& i\partial_{t} \Phi_{1}+\partial^{2}_{x} \Phi_{1}+2(|\Phi_{1}|^2
+2|\Phi_{0}|^2) \Phi_{1} -2\Phi_{1}^{*}\Phi_{0}^2=0,  \\
& i\partial_{t} \Phi_{0}+\partial^{2}_{x}\Phi_{0}+2(2|\Phi_{1}|^2
+|\Phi_{0}|^2) \Phi_{0} -2\Phi_{0}^{*}\Phi_{1}^2 =0,
\end{aligned}
\end{equation}
If we insert eq.  (\ref{eq:2MNLS'}) into (\ref{eq:fi1}) we obtain
\begin{equation}\label{eq:w12}
\begin{aligned}
& i\partial_{t} w_{1}+\partial^{2}_{x} w_{1}+ 2|w_{1}|^2 w_{1} =0,  \\
& i\partial_{t} w_{2}+\partial^{2}_{x}w_{2}+2|w_{2}|^2 w_{2}  =0,
\end{aligned}
\end{equation}
Therefore, if we want to analyze the specific features of $F=1$ BEC we have to avoid such initial conditions.

Similarly, if for $n=5$ we choose  in (\ref{eq:1s-3}) $\nu_{01;1}= \nu_{01,5}$,  $\nu_{01;2}= -\nu_{01,4}$ we will obtain
in fact a solution to $F=1$ BEC.

\section{Conclusions and discussion}

Using the Zakharov-Shabat dressing method we have obtained the two-soliton solution and
have used it to analyze the soliton interactions of the MNLS equation. The conclusion is that
after the interactions the solitons recover their polarization vectors $\nu_{0k}$, velocities
and frequency velocities. The effect of the interaction  is, like in for the scalar NLS equation,
shift of the center of mass $z_1\to z_1 +r_+$ and shift of the phase $\phi_1 \to \phi_1 + \alpha_+$.
Both shifts are expressed through the related eigenvalues $\lambda_j^\pm$ only.

The next step would be to analyze multi-soliton interactions. Our hypothesis is that each soliton
will acquire a total shift of the center of mass that is sum of all elementary shifts from each
two soliton interactions. Similar result is expected for the total phase shift of the soliton.

Finally we have proved a theorem, stating that a symmetry imposed on the minimal set of scattering
data leads to a symmetry of the corresponding solution. So
if we want to analyze the specific features of a given MNLS   we have to avoid such initial conditions.

\section*{Acknowledgements} I am grateful to Professors J. C. Maraver, P. G. Kevrekidis and R.
Carretero-Gonzales for giving me the chance to  take part in this proceedings. I also wish to thank Professor N. Kostov
and Dr. T. Valchev for useful discussions and the referee for useful remarks.

\end{document}